\documentclass[reqno,a4paper]{article}
\usepackage[utf8]{inputenc}
\usepackage[english]{babel}
\usepackage{cite}

\def\A{{\mathcal{A}}}
\def\X{{\mathcal{X}}}
\def\E{{\mathbb{E}}}

\def\C{{\mathcal{C}}}
\def\I{{\mathcal{I}}}
\def\T{{\mathcal{T}}}
\def\reals{{\mathbb{R}}}
\usepackage{comment}
\usepackage{amsmath,amsfonts,amssymb,amsthm}
\usepackage{graphicx} 
\usepackage{float}
\usepackage{amsmath}
\graphicspath{ {images/} }

\newtheorem{theorem}{Theorem}
\newtheorem{lemma}[theorem]{Lemma}
\newtheorem{proposition}[theorem]{Proposition}

\usepackage{lipsum,eso-pic} 

\usepackage{lipsum,eso-pic} 

\usepackage{fancyhdr} 

\usepackage{background}
\usepackage{todonotes}

\graphicspath{ {images/} }
\backgroundsetup{
	scale=0.6,
	angle=0,
	opacity=0.5,
	color=white,
	contents={\begin{tikzpicture}[remember picture,overlay]
		\node at ([xshift=2in,yshift=3in] current page.north east) 
		{\rule{1cm}{1cm}
			NEED TO FIND LOGO
		}; 
		\end{tikzpicture}}
}

\usepackage[scale=0.7]{geometry}
\usepackage{graphicx,color,amsthm,hyperref,amsmath,amssymb}
\usepackage[utf8]{inputenc}
\usepackage[numbers]{natbib}
\usepackage{setspace}

\def\reals{{\mathbb R}}

\newtheorem{conjecture}{Conjecture}

\theoremstyle{definition}

\theoremstyle{remark}

\newcommand{\horrule}[1]{\rule{\linewidth}{#1}} 
\title{	
	\normalfont \normalsize
	\textsc{BEN- GURION UNIVERSITY OF THE NEGEV} \\ [25pt] 
	\textsc{THE FACULTY OF NATURAL SCIENCES}\\[25pt]
	\textsc{DEPARTMENT OF COMPUTER SCIENCE}\\[25pt]
	\horrule{0.5pt} \\[0.4cm] 
	\huge A Crossing Lemma for Families of Jordan Curves with a Bounded Intersection Number \\  
	\horrule{2pt} \\[0.5cm] 
\date{}
}

\begin{document}
		\title{	
		A Crossing Lemma for Families of Jordan Curves with a Bounded Intersection Number
	}
\author{Maya Bechler-Speicher\\ \\Ben-Gurion University}

\date{}

	\maketitle

		\begin{abstract}
 A family of closed simple (i.e., Jordan) curves is {\it $m$-intersecting} if any pair of its curves have at most $m$ points of common intersection. We say that a pair of such curves {\it touch} if they intersect at a single point of common tangency.
In this work we show that any $m$-intersecting family of $n$ Jordan curves in general position in the plane contains $O\left(n^{2-\frac{1}{3m+15}}\right)$ touching pairs.\footnote{A family of Jordan curves is in general position if no three of its curves pass through the same point, and no two of them overlap. The constant of proportionality with the $O(\cdot)$-notation may depend on $m$.}

Furthermore, we use the string separator theorem of Fox and Pach \cite{FP10} in order to establish the following Crossing Lemma for contact graphs of Jordan curves: Let $\Gamma$ be an $m$-intersecting family of closed Jordan curves in general position in the plane with exactly $T=\Omega(n)$ touching pairs of curves,  then the curves of $\Gamma$ determine $\Omega\left(T\cdot\left(\frac{T}{n}\right)^{\frac{1}{9m+45}}\right)$ intersection points.

This extends the similar bounds that were previously established by Salazar for the special case of pairwise intersecting (and $m$-intersecting) curves. Specializing to the case at hand, this substantially improves the bounds that were recently derived by Pach, Rubin and Tardos for arbitrary families of Jordan curves.
	\end{abstract}
		
		\clearpage
		\maketitle

		\section{Introduction and Main Results}
\paragraph{Notation: intersection patterns of  closed Jordan curves and Jordan arcs.} 
A {\it Jordan curve} $\gamma$ is a simple  (i.e., non-self-intersecting) closed curve in the Euclidean plane; by the fundamental Jordan Curve Theorem, its complement $\reals^2\setminus \gamma$ consists of two simply connected planar regions. If a pair of Jordan curves have precisely one intersection point then they are tangent at this point, in which case we say 
that the curves {\it touch} and refer to their intersection as a $\mathit{touching  }$  $\mathit{ point}$; see Figure \ref{Fig:TouchingCrossing}.

\begin{figure}[htbp]
\begin{center}
\includegraphics[scale=0.4]{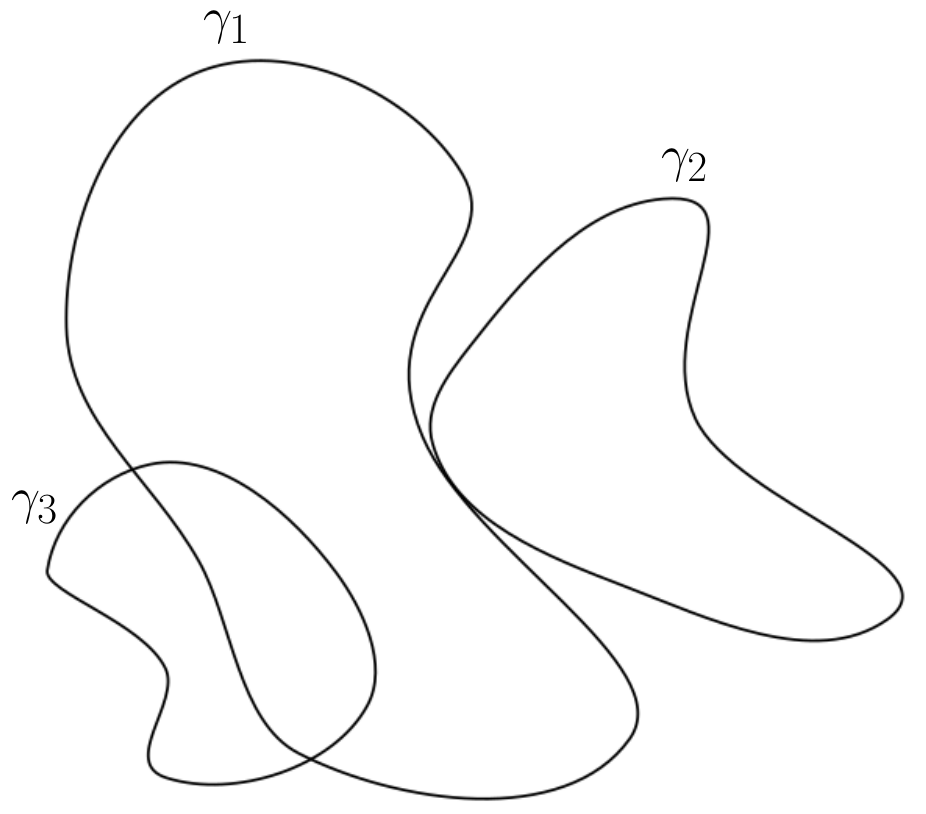}\hspace{2cm} \includegraphics[scale=0.4]{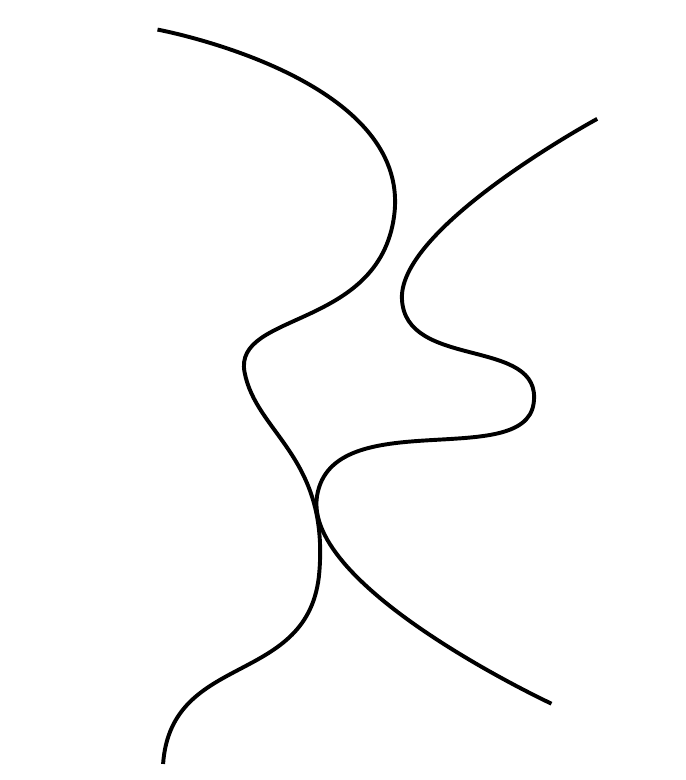}
\caption{\small  Left: The pair of Jordan curves $\gamma_1$ and $\gamma_2$ touch at a single point, whereas the pair $\gamma_1$ and $\gamma_3$ cross without touching. Right: A touching pair of Jordan arcs, whose single intersection point is their point of common tangency.}
\label{Fig:TouchingCrossing}
\end{center}
\end{figure}

For the sake of uniformity, most of our results in the sequel are established for the more general case of {\it Jordan arcs} -- simple connected (and compact) arcs in the plane. In particular, every Jordan curve is a Jordan arc but not vice versa.
To this end, we say that a pair of Jordan arcs {\it touch} if they intersect at a single point of common tangency.\footnote{In contrast with the case of closed Jordan curves, it is possible that a pair of Jordan arcs have a single intersection point at which they cross transversally.} See Figure \ref{Fig:TouchingCrossing} (right).

A family $\Gamma$ of Jordan arcs is in {\it general position} if no three of its curves pass through the same point, no two of them overlap, and none of them passes through an endpoint of another arc.
We say that such a family $\Gamma$ is $m$-intersecting if any two curves in $\Gamma$ intersect in at most $m$ points.

\paragraph{Intersection and contact graphs.} Let $\Gamma$ be a family of Jordan arcs. The {\it intersection}, or the {\it string} graph $\I(\Gamma)$ of $\Gamma$  defined over the vertex set $\Gamma$, describes pairwise intersections among the curves of $\Gamma$: every vertex in $\I(\Gamma)$ corresponds to an arc of $\Gamma$, and a pair of such vertices are connected by an edge in $\C(\Gamma)$ if and only if their respective arcs intersect; see Figure \ref{Fig:StringContact}.
Similarly, the contact (or touching) graph $\C(\Gamma)$ of $\Gamma$ is defined similarly but describes all pairwise touchings amongst the curves of $\Gamma$.

\begin{figure}[htbp]
\begin{center}
\includegraphics[scale=0.6]{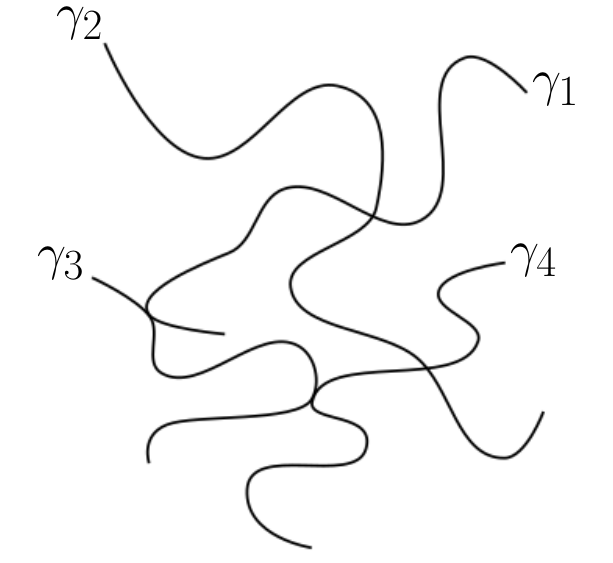}\hspace{1cm}	\includegraphics[scale=0.7]{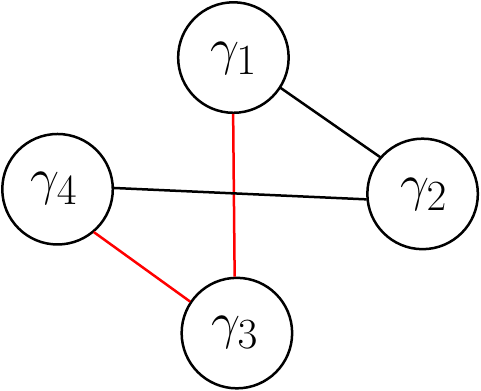}
\caption{\small  A family of Jordan arcs (left) with its intersection graph $\I(\Gamma)$ (right). The red edges represent touchings and belong to the contact graph $\C(\Gamma)$.}
\label{Fig:StringContact}
\end{center}
\end{figure}
		
		The study of string and contact graphs of Jordan curves and Jordan arcs in the plane is a recurring theme in combinatorial and computational geometry \cite{PseudoCircles,AgS05,ArS02,CEGSW90,PRT15,
PRT16,PRT18,PST12,Sa99,TT98} due to its numerous applications which range from long standing Erd\H{o}s conjectures in plane geometry \cite{De98,Er46} to robotic motion planning \cite{KLPS86,ShA95}. For example, a famous open problem due to Erd\H{o}s \cite{Er46} concerns the largest possible number of unit distances that can simultaneously arise among $n$ points in the plane can be reformulated as follows: 
		
\smallskip		
		{\it What is the maximum possible number of touchings can be determined by a family of $n$ unit circles in the plane?}
	
	\medskip
It is immediate to check that the contact graph of any family of Jordan curves is planar;
moreover, a celebrated and very deep result due to Koebe \cite{Koe36}, Andreev \cite{An70} and Thurston \cite{Thu97} states that any planar graph can be realized as (a subset of) a contact graph of a certain family of discs.

\paragraph{A Crossing Lemma for Jordan curves?} If the graph $G=(V,E)$ is not planar, the cornerstone Crossing Lemma in topological graph theory, due to Ajtai {\it et al.} \cite{ACNS} and Leighton \cite{L}, provides a lower bound of $\Omega(|E|^3/|V|^2)$ on the number of edge intersections in any planar embedding of the graph. In particular, if the number of edges $|E|$ exceeds the number $|V|$ of vertices, then the number of edge intersections is $\omega(|E|)$. It is natural to ask if a comparable relation exists for contact graphs of Jordan curves?

 The answer is intimately related to the following long standing conjecture by Richter and Thomassen \cite{RiT95}: Any family of $n$ pairwise intersecting Jordan curves in general position determines at least $n^2(1-o(1))$ intersection points.
The conjecture was recently confirmed by Pach, Rubin and Tardos \cite{PRT16,PRT18}, who established the following loose analogue Crossing Lemma for contact graphs: 

\begin{theorem}[Pach-Rubin-Tardos, 2018]\label{Thm:PRT}

(i) Let $\Gamma$ be a family of $n$ pairwise intersecting closed Jordan curves in general position in the plane that determine $T$ touching pairs. Then the curves of $\Gamma$ determine $\omega(T)$ intersection points. 

(ii) There is a constant $c_{1}>0$ so that the following property holds:
 Let $\Gamma$ be a family of $n$ pairwise intersecting closed Jordan curves in general position in the plane with at least $T\geq n$ touching pairs, then $\Gamma$ determines $\Omega\left(T\log\log^{c_{1}}\left(\frac{T}{n}\right)\right)$ intersection points.
\end{theorem}

Specializing to the $m$-intersecting families $\Gamma$ of Jordan curves, the first part of Theorem \ref{Thm:PRT} implies that the number $T$ of the touching pairs is only $o(n^2)$. 
In 1999 Salazar \cite{Sa99} confirmed the Richter-Thomassen Conjecture for $m$-intersecting families of Jordan curves. To that end, he used K\H{o}v\'{a}ri-S\'{o}s-Tur\'{a}n Theorem \cite{KST54} from extremal graph theory in order to establish an even stronger variant of Theorem \ref{Thm:PRT}.

\begin{theorem}[Salazar, 1995]\label{Thm:Salazar}
Let $m>1$ be an integer. Then there is $C_{Sal}>0$ with the following property: Any $m$-intersecting family of $n$ pairwise intersecting (closed) Jordan curves in general position in the plane determines at most $\displaystyle O\left(n^{2-\frac{C_{Sal}}{m}}\right)$ touchings.\footnote{In the sequel, the constant of proportionality within the $O\left(\cdot \right)$ notation depends on the fixed constant parameter $m$.}
\end{theorem} 

If the family $\Gamma$ in Theorem \ref{Thm:Salazar} is $2$-intersecting, it is called a family of {\it pseudo-circles}. Agarwal {\it et al.} \cite{PseudoCircles} showed that any such pairwise-intersecting family of pseudo-discs in general position determines only $O(n)$ touchings.

Unfortunately, the analysis of Salazar (and also of Agarwal {\it et al.}) critically relies on the assumption that the Jordan curves are pairwise intersecting and, therefore, falls short of providing a proper Crossing Lemma for $m$-intersecting families of Jordan curves.

Note that {\it any} family $\Gamma$ of non-overlapping bounded-degree {\it algebraic} curves (i.e., connected components of $1$-dimensional algebraic varieties within $\reals^2$) is in particular $m$-intersecting, where the constant $m$ depends on the maximum degree of the curves within $\Gamma$.
A recent result by Ellenberg, Solymosi and Zahl \cite{ESZ16} implies that such  a family of $n$ curves determines $O\left(n^{3/2}\right)$ pairwise tangencies, where only the constant proportionality depends on the maximum degree. Unfortunately, their bound essentially relies on powerful tools from algebraic geometry and does not extend to more general instances of $m$-intersecting families of Jordan curves.

\paragraph{Our contribution.} In this thesis we extend Theorem \ref{Thm:Salazar} to general $m$-intersecting families of Jordan arcs in general position.

\begin{theorem}	\label{Thm:Main1}
Let $m\geq 1$ be an integer then any $m$-intersecting family of $n$ Jordan arcs in general position in the plane determines $O\left(n^{2-\frac{1}{3m+15}}\right)$ touchings; here the constant of proportionality may depend on $m$.
	\end{theorem}
	
Note that Theorem \ref{Thm:Main1} can be viewed as a Crossing Lemma for dense contact graphs of Jordan curves. We then employ the machinery of string separators \cite{FP10} in order to derive a more general Crossing Lemma for arbitrary $m$-intersecting families of Jordan arcs. 
	
\begin{theorem}	\label{Thm:Main2}
Let $m\geq 1$. Any $m$-intersecting family of $n$ Jordan arcs in general position with $T\geq n$ touching pairs determines\\ $\Omega \left(T\cdot\left(\frac{T}{n}\right)^{\frac{1}{9m+45}}\right)$ intersection points, where the constant of proportionality may depend on $m$.
	\end{theorem}

Specializing to $m$-intersecting families of Jordan curves, the bounds in Theorems \ref{Thm:Main1} and \ref{Thm:Main2} constitute a substantial improvement of the lower bound in Theorem \ref{Thm:PRT}.

\paragraph{Organization and overview.} The thesis is organized as follows.
In Section \ref{Sec:Prelim} we introduce the basic notions from combinatorial geometry and extremal combinatorics that are used throughout the thesis, and state the key properties that lie at the heart of our analysis. In Section \ref{Sec:Main1} we establish Theorem \ref{Thm:Main1} through a careful combination of K\H{o}v\'{a}ri-S\'{o}s-Tur\'{a}n Theorem \cite{KST54} with a random sampling argument which is inspired by the analysis of Pach, Rubin and Tardos \cite{PRT18}. In Section \ref{Sec:Main2} we establish the general Crossing Lemma -- Theorem \ref{Thm:Main2}. Informally, this is achieved by decomposing the arrangement of the Jordan curves (cf. for the precise definition) into pieces which determine sufficiently dense intersection graphs, and then applying Theorem \ref{Thm:Main1} to each of these sub-instances.

\section{Preliminaries}\label{Sec:Prelim} 

\paragraph{Arrangements of Jordan curves and Jordan arcs in the plane.} 
Any family $\Gamma$ of closed Jordan arcs in general position yields a decomposition of $\reals^2$ which is called an {\it arrangement} $\A(\Gamma)$. The {\it faces}, or {\it cells}, of $\A(\Gamma)$ are maximal connected regions of $\reals^2\setminus\left(\bigcup_{\gamma\in \Gamma}\gamma\right)$; see Figure \ref{Fig:Arrangement}. The vertices of $\A(\Gamma)$ are intersection points amongst the arcs of $\Gamma$, or endpoints of the arcs.
For every face of $\A(\Gamma)$, every boundary component of it is composed of connected portions of the curves of $\Gamma$, which are called the {\it edges} of $\A(\Gamma)$. Notice that some of these edges may be adjacent to less than two vertices of $\A(\Gamma)$, in which case their closure coincides with a closed curve of $\Gamma$.

\begin{figure}[htbp]
\begin{center}
\includegraphics[scale=0.4]{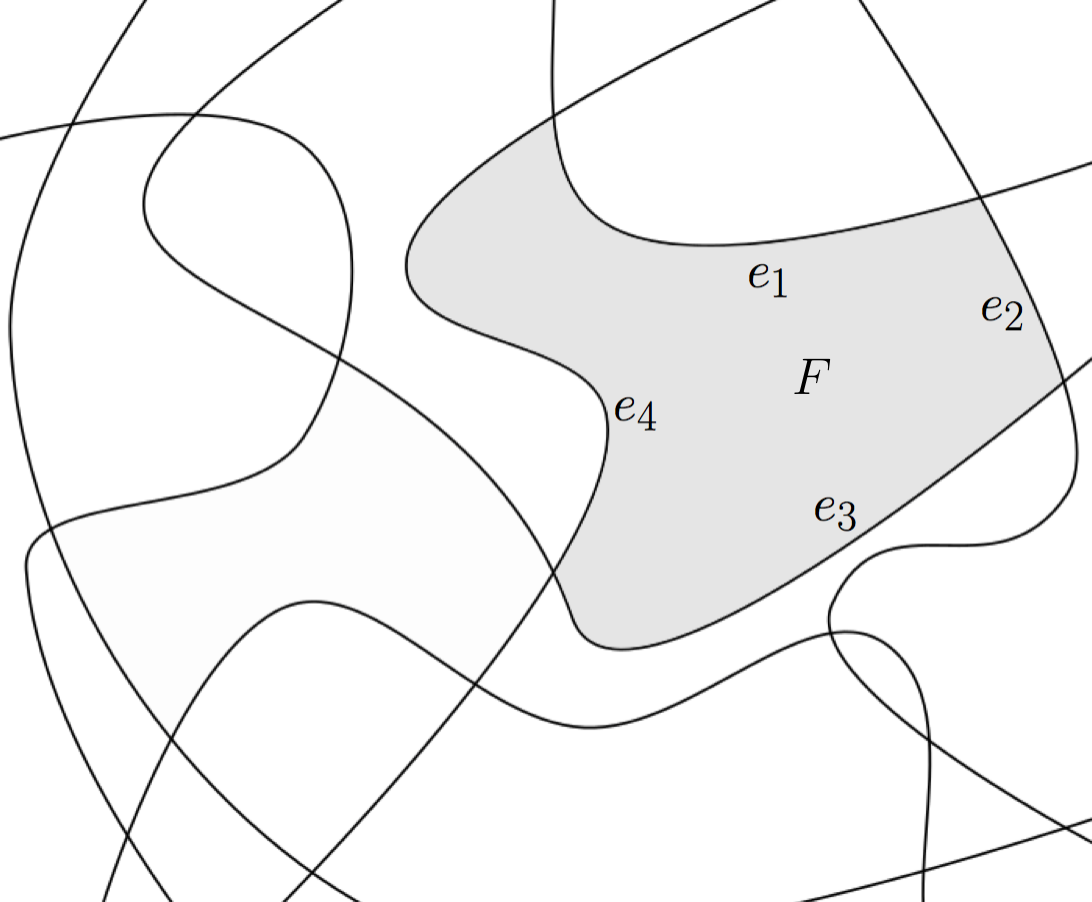}
\caption{\small An arrangement of Jordan curves. One of the faces $F$ is shaded. It is bounded by the $4$ edges $e_1,\ldots,e_4$.}
\label{Fig:Arrangement}
\end{center}
\end{figure}

If the family $\Gamma$ is $m$-intersecting, then both the number of vertices and the number of faces in $\A(\Gamma)$ are proportional to the number of edges in the intersection graph $\I(\Gamma)$ (where the constant of proportionality may depend on $m$), and their overall number is $O(n^2)$.
We refer the reader to a popular book \cite{ShA95} for further discussion of curve arrangements and their higher-dimensional analogues, and their numerous applications in computational geometry.

\paragraph{A Separator Theorem for intersection graphs.}
Given a graph $G=(V,E)$ of order $n=|V|$, a subset $V'\subset V$ of vertices is called a {\it separator} if each connected component of $G[V\setminus V']$ contains at most $2n/3$ of the vertices.

The celebrated Separator Theorem of Lipton and Tarjan \cite{LiTa79} states that any planar graph admits a separator of cardinality $O\left(\sqrt{n}\right)$.
Their result has been recently extended, by Fox and Pach, to non-dense string graphs.

\begin{theorem}[Fox-Pach, 2010 \cite{FP10}]\label{Thm:FoxPach}
Let $\Gamma$ be a family Jordan arcs in general position in the plane. Then the intersection graph $\I(\Gamma)$ admits a separator $S$ of cardinality $O\left(\sqrt{x}\right)$, where $x$ denotes the number of intersection points amongst the arcs of $\Gamma$. In particular, if the family $\Gamma$ is $m$-intersecting, for $m>1$, then we have $|S|=O\left(\sqrt{mE}\right)$; here $E$ denotes the number of edges in the intersection graph $\I(\Gamma)$. See Figure \ref{Fig:Separator}.
\end{theorem}

\begin{figure}[htbp]
\begin{center}
\includegraphics[scale=0.3]{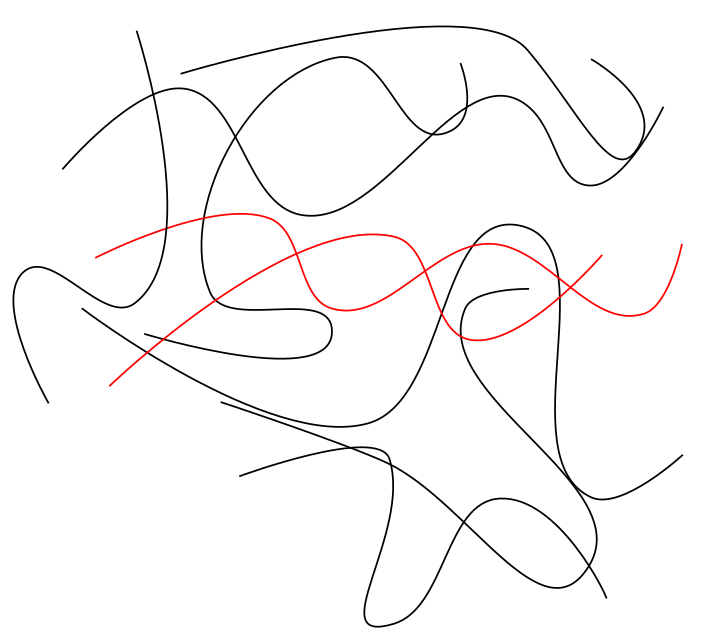}
\caption{\small  A $4$-intersecting family $\Gamma$ of Jordan arcs, including separator set of arcs (marked red).}
\label{Fig:Separator}
\end{center}
\end{figure}

Informally, the proof of Theorem \ref{Thm:FoxPach} proceeds by converting the arrangement $\A(\Gamma)$ of $\Gamma$ into a planar graph and a careful application of (a weighted variant of) the Lipton-Tarjan Theorem.

\paragraph{K\H{o}v\'{a}ri-S\'{o}s-Tur\'{a}n Theorem.} In his proof of Theorem \ref{Thm:Salazar}, Salazar used the following key result from extremal combinatorics.

\begin{theorem}[K\H{o}v\'{a}ri-S\'{o}s-Tur\'{a}n, 1954 \cite{KST54}] \label{KST}
Let $s,t$ be positive integer constants, and let  $G=(V,E)$ be a graph which does not contain a copy of the complete $s$-by-$t$ bipartite graph $K_{s,t}$. Then we have $|E|=O\left(n^{2-1/s}\right)$ edges, where the constant of proportionality depends on $t$.
\end{theorem}

To establish Theorem \ref{Thm:Salazar}, Salazar uses the $m$-intersecting property of $\Gamma$ to show that the contact graph $\C(\Gamma)$ cannot contain a copy of $K_{2m,l}$ for some suitable constant $l>0$. Unfortunately, his analysis does not directly extend to the $m$-intersecting families $\Gamma$ in which some pairs of the curves are disjoint.

		\section{Proof of Theorem \ref{Thm:Main1}}\label{Sec:Main1}
		
	Here is an informal overview of our proof of Theorem \ref{Thm:Main1}. 
The crux of Salazar's analysis for pairwise intersecting families $\Gamma$ of Jordan curves is that any sufficiently dense contact graph $\C(\Gamma)$ must contain a large bi-clique $K_{2m,l}$, for some $l=\Omega(m^3)$. In such a case, we can show a contradiction by observing that some pair of curves in the $m$-intersecting family $\Gamma$ have to intersect more than $m$ times. However, this geometric argument fails if some pairs of curves are disjoint.

Fortunately, the analysis of Salazar applies if there exist disjoint families $A,B\subset \Gamma$ so that all of the curves in $A$ (resp., $B$) touch the same curve $\gamma_1\in \Gamma$ (resp., $\gamma_2\in \Gamma$), and sufficiently many of the touchings involve a curve of $A$ and a curve of $B$.
The following lemma shows that the desired pair $\gamma_1,\gamma_2\in \Gamma$ always exists (see Figure 1).
	
\begin{lemma}\label{Lemma:GroundCurves}
For every $0<C_{lem}$, and every integer  $m>0$, there is an integer $n_0\geq 2$ with the following property.

Let $\Gamma$ be an $m$-intersecting family of $n\geq n_0$ Jordan arcs in general position and let $\T$ denote the set of all the touching points in $\Gamma$. If $|\T|\geq C_{lem}n^{2-\frac{1}{3m+15}}$, then there exists a pair of arcs $\gamma_1, \gamma_2 \in \Gamma$, a cell $\Delta \subseteq \reals^2\setminus (\gamma_1 \cup \gamma_2)$, and disjoint subsets $A,B\subset \Gamma$,  with the following properties:\footnote{The constants of proportionality within the $O\left(\cdot\right)$-notation may depend on the choice of $C_{lem}$ and $m$.}

(1) every arc in $A$(resp., $B$) touches $\gamma_1$ (resp., $\gamma_2$), and 

(2) $\Omega \left(\frac{|\T|}{n^{\frac{2}{3m+15}}}\right)$ of the touching points in $\T$ lie in $ \Delta$ and involve an arc of $A$ and an arc of $B$.
	\end{lemma}

	\begin{figure}[htbp]
\begin{center}
\includegraphics[scale=0.4]{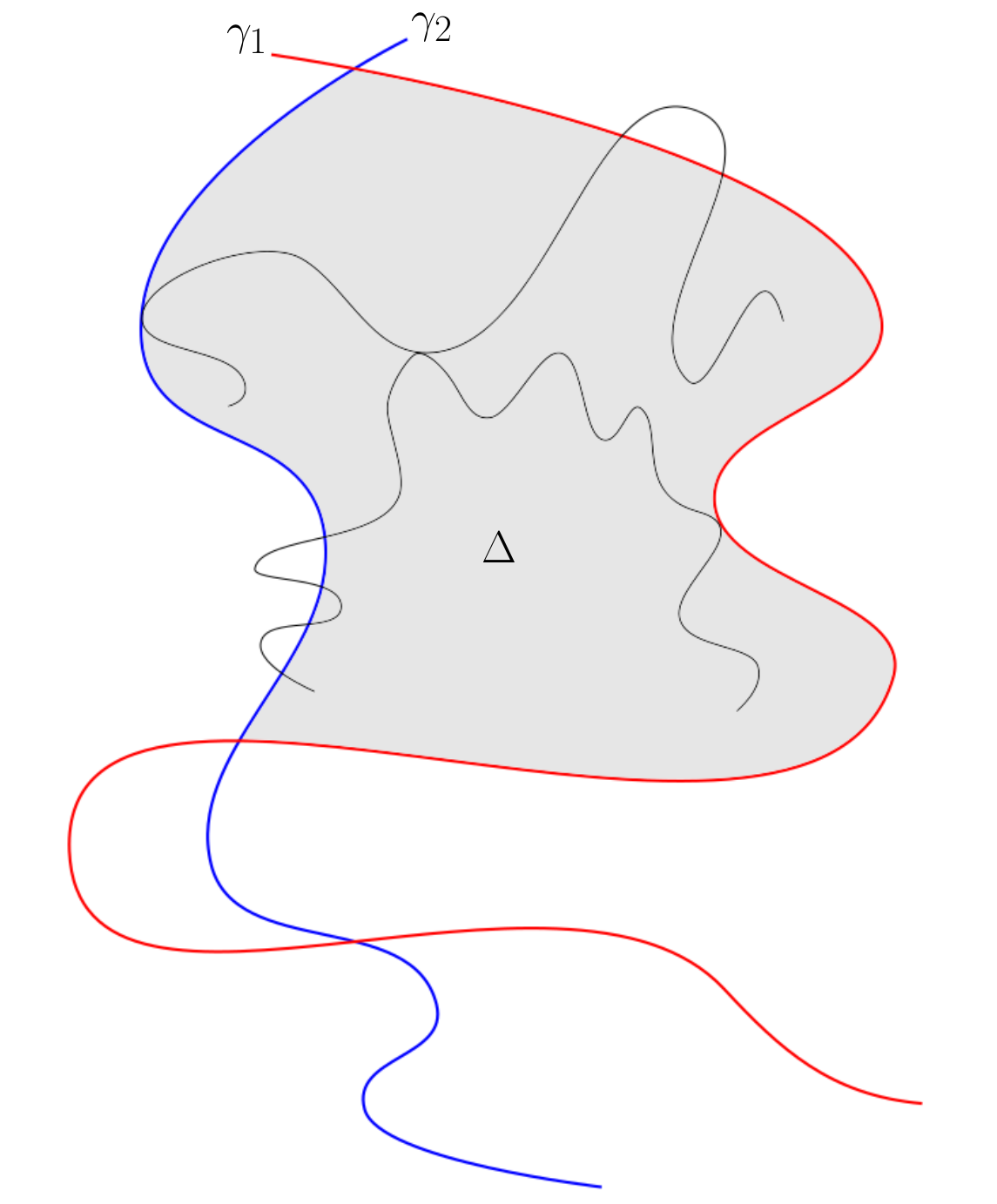}
\caption{\small Lemma \ref{Lemma:GroundCurves}. We seek a pair of arcs $\gamma_1,\gamma_2\in \Gamma$, a cell $\Delta \subseteq R\setminus (\gamma_1 \cup \gamma_2)$, and disjoint subfamilies $A,B\subset \Gamma$, so that every arc in $A$ (resp., $B$) touches $\gamma_1$ (resp., $\gamma_2$), and $\Omega \left(\frac{|\T|}{n^{\frac{2}{3m+15}}}\right)$ of the touchings amongst $\Gamma$ involve an arc in $A$ and another arc in $B$.}
\label{Fig:Grounds}
\end{center}
\end{figure}

	\begin{proof} 
Let $\Gamma$ be a family of $n$ Jordan arcs in general position whose induced set of touching points $\T$ satisfies $|\T|\geq C_{lem}n^{2-\frac{1}{3m+15}}$.
We can assume that $n\geq n_0$, where the choice of $n_0$ is detailed in the sequel.

Consider an arbitrary labeling $\gamma_1,\ldots,\gamma_n$ of the curves in $\Gamma$.
We select uniformly and at random a pair of distinct arcs in $\Gamma$ (so that each pair is chosen with uniform probability $1/{n\choose 2}$). Assume with no loss of generality that we have chosen $\gamma_1$ and $\gamma_2$. Let $A'$ (resp., $B')$ be the set of all the arcs in $\Gamma\setminus \{\gamma_1,\gamma_2\}$ that are touching $\gamma_1$ (resp., $\gamma_2)$.  
We first set $X:= A'\cap B'$ and assign every arc in $A'\setminus B'$ and $B'\setminus A'$ to the respective set $A$ or $B$.  We then assign every arc in $X$ exclusively, independently and at random, with probability $1/2$, to either of the sets $A$, $B$. Clearly, this yields disjoint subfamilies $A,B\subset \Gamma$ so that $A\uplus B=A'\cup B'$ and every arc in $A$ (resp. $B$) touches $\gamma_1$ (resp. $\gamma_2$). 

Let $\T^*$ be the set of all the touchings that involve an arc in $A$ and another arc in $B$, and let $\Delta$ be the open cell of $\reals^2\setminus (\gamma_1\cup \gamma_2)$ that contains the largest number of touchings of $\T^*$ (if $\Delta$ is not unique, we select any such cell). Since $\reals^2\setminus(\gamma_1\cup \gamma_2)$ has at most $m+2$ cells, the restricted set of touchings $\T^*\cap \Delta$ has cardinality at least $|\T^*|/(m+2)$. 

It suffices to show that the expectation of $|\T^*|$  is at least $\Omega \left(\frac{|\T|}{n^{\frac{2}{3m+15}}}\right)$.
Indeed, let $\T'$ be the set of all the touchings that involve an arc $\gamma_i\in A'$ and another arc $\gamma_j\in B'$. Clearly, we have $\T'\supset \T^*$ and every touching of $\T'$ is included in $\T^*$ with probability at least $1/2$. Therefore, we have $\E[|\T^*|]\geq \frac{\E[|\T'|]}{2}$.

To establish a lower bound for $\E[|\T'|]$ in terms of $|\T|$, we say that an arc $\gamma \in \Gamma$ is $\mathit{ poor}$ if it contains fewer than $\frac{|\T|}{1000n}$ touchings of $\T$. Let $\T_{poor}\subseteq \T$ denote the set of all the touching points that lie on at least one poor arc. It follows that $|\T_{poor}|\leq \frac{|\T|}{1000}$. We thus denote $\Gamma_{rich}:=\Gamma\setminus \Gamma_{poor}$ and say that a touching point $t\in \T$ is {\it rich} if it does not belong to $\T_{poor}$; that is, both of its adjacent arcs belong to $\Gamma_{rich}$. Note that the subset $\T_{rich}$ of all such rich touchings within $\T$ has cardinality at least $\frac{999}{1000}|\T|$.

Fix a rich touching point $t\in \T_{rich}$. Let $\gamma_i,\gamma_j\in \Gamma_{rich}$ be the arcs that are adjacent to $t$. Clearly, $t$ is included in $\T'$ if and only if (1) $i,j\geq 3$, and (2) the arc $\gamma_i$ touches $\gamma_1$ and the arc $\gamma_j$ touches $\gamma_2$ (or vice versa). 
Since $t$ is rich, each of $\gamma_i,\gamma_j$ belongs to $\Gamma_{rich}$ and, therefore, touches at least $C_{lem} n^{1-\frac{1}{3m+15}}$ arcs of $\Gamma\setminus \{\gamma_i,\gamma_j\}$. Thus, there exist at least $(C^2_{lem}/2)n^{2-\frac{2}{3m+15}}-3n/2$ favourable pairs $\{\gamma_1,\gamma_2\}$ (whose selection secures $t$ in $\T'$). A sufficiently large choice of $n_0$ and $n\geq n_0$ guarantees that $t$ is included in $\T'$ with probability at least 

$$
 \frac{C^2_{lem}n^{2-\frac{2}{3m+15}}-3n}{2{n\choose 2}}=\Omega\left(\frac{1}{n^{\frac{2}{3m+15}}}\right).
 $$

Using the linearity of expectation, we obtain

$$
\E[|\T'|]=\Omega\left(\frac{|\T|}{n^{\frac{2}{3m+15}}}\right),
$$

\noindent and readily conclude that 

$$
E\left(|\T^*\cap \Delta|\right)\geq \frac{\E\left(|\T^*|\right)}{m+2}\geq \frac{\E\left(|\T'|\right)}{2(m+2)}=\Omega\left(\frac{|\T|}{2(m+2)n^{\frac{2}{3m+15}}}\right)=
$$
$$
=\Omega\left(\frac{|\T|}{n^{\frac{2}{3m+15}}}\right).
$$

This completes the proof Lemma \ref{Lemma:GroundCurves}.
\end{proof}

\noindent{\bf Back to the proof of Theorem \ref{Thm:Main1}.} 
We choose a suitably small, albeit fixed, $0<C_{lem}$, and a suitably large $n_0$ as in Lemma \ref{Lemma:GroundCurves}.
Let $\T$ be the set of all the touchings that are attained by $\Gamma$. Denote $T=|\T|$. We may assume that $T\geq C_{lem} n^{2-\frac{1}{3m+15}}$, and $n\geq n_0$, or else the theorem immediately follows.
Lemma \ref{Lemma:GroundCurves} yields a pair of arcs $\gamma_1, \gamma_2\in \Gamma$, a cell $\Delta\subset \reals^2\setminus (\gamma_1\cup \gamma_2)$ in the arrangement of $\gamma_1$ and $\gamma_2$, and disjoint subsets $A,B\subset \Gamma$ of arcs that touch, respectively, $\gamma_1$ and $\gamma_2$ and determine $\Omega \left(\frac{T}{n^{\frac{2}{3m+15}}}\right)$ touchings within $\Delta$. 

We now split each arc in $A\cup B$ into at most $m+2$ sub-arcs by cutting it at the at most $m+1$ intersection points with $\gamma_1\cup\gamma_2$.\footnote{Recall that every arc in $A$ (resp., $B$) touches $\gamma_1$ (resp., $\gamma_2$).} Let $\Lambda_{\Delta}$ denote the collection of all these subarcs that lie within $\Delta$. Notice that the endpoints of every arc of $\Lambda_{\Delta}$ lie on the boundary of $\Delta$, with the possible exception of two of them; see Figure \ref{Fig:CutArcs}.

\begin{figure}[htbp]
\begin{center}
\includegraphics[scale=0.5]{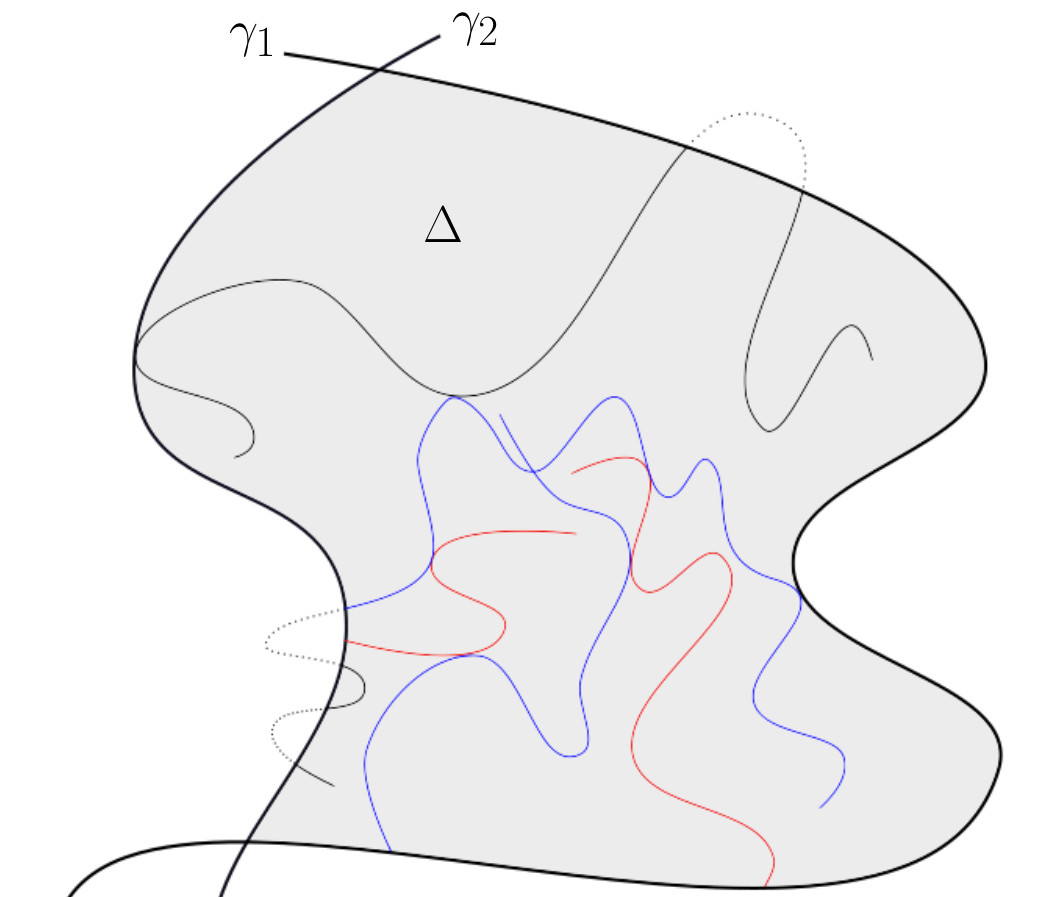}
\caption{\small  The cell $\Delta$ (shaded) with the families of subarcs $\Lambda_1$ (red) and $\Lambda_2$ (blue).}
\label{Fig:CutArcs}
\end{center}
\end{figure}

\begin{lemma}\label{Lemma:NoBiClique}
		With the above definitions, there is a constant $l=l(m) \geq 1$  (which depends only on $m$) so that the contact graph $\C(\Lambda_\Delta)$ of $\Lambda_\Delta$ cannot contain $K_{m+5,l(m)}$. 
		\end{lemma}
	
		\begin{proof}
	Let $K_{m+5,l}$ be a bipartite sub-graph in $\C(\Lambda_{\Delta})$.
	Namely, there exist disjoint subsets $\Lambda_1,\Lambda_2\subset \Lambda_\Delta$ with $|\Lambda_1|=m+5$ and $|\Lambda_2|=l$ so that every arc of $\Lambda_1$ touches every arc of $\Lambda_2$. In what follows we show that $l$ is bounded by some constant $l(m) > 0$. 
To that end, we examine the arrangement  $\A_\Delta$ of $\Lambda_1\cup \{\gamma_1,\gamma_2\}$. Clearly, every arc in $\Lambda_2$ is fully contained in the closure of a single face of this arrangement. 
Since the family $\Gamma$ is $m$-intersecting, this is also true for $\Lambda_{\Delta}$, so the arrangement $\A_\Delta$ has at most $O\left(m^3\right)$ faces, and the boundary of each of these faces consists of $O(m^3)$ edges. Therefore, it is enough to bound the number of arcs in $\Lambda_2$ that are contained in (the closure of) any given face $F$ in $\A_\Delta$.

Indeed, let $F\subset \Delta$ be such a cell whose closure contains $z>0$ arcs of $\Lambda_2$, and denote the subset of these arcs by $\Lambda_\Delta$. Since every arc of $\Lambda_1$ meets the boundary of $\Delta$, the face $F$ is simply connected. (In other words, its boundary consists of a single connected component.)

\begin{figure}[htbp]
\begin{center}
\includegraphics[scale=0.8]{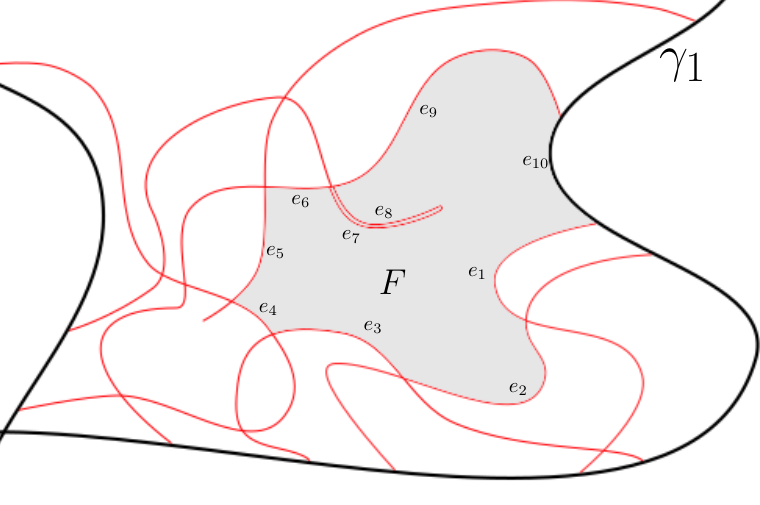}
\caption{\small  Proof of Lemma \ref{Lemma:NoBiClique}. The simply connected cell $F\subset \Delta$ in the arrangement of $\Lambda_1\cup\{\gamma_1,\gamma_2\}$. The edges $e_1,\ldots, e_{10}$ appear in this order along $\partial F$. (Notice that the edges $e_7$ and $e_8$ are two copies of the same edge. However, any given arc of $\Lambda_F$ can touch only one of them.)}
\label{Fig:CutArcs}
\end{center}
\end{figure}

We trace the boundary of $F$ so that its interior remains to the right of us. Notice that every edge of $\partial F$ is contained in an arc of $\Lambda_1\cup\{\gamma_1,\gamma_2\}$, and some of them may be encountered twice, from both sides. (This can happen only if their respective arcs have an endpoint in the interior of $\Delta$.) In this case we use distinct labels for the two sides of the edge.
We thus list the edges $\left<e_1,\ldots,e_h,e_1\ldots\right>$ of $\partial F$ in the order of their appearance during this walk.
The crucial observation is that every arc $\lambda\in \Lambda_F$ touches every arc of $\Lambda_1$ at a unique boundary edge of $F$. Hence, each $\lambda\in \Lambda_F$ determines a unique sub-sequence $\sigma(\lambda)=\left<e_{j_1}(\lambda),e_{j_2}(\lambda),\ldots,e_{j_{m+5}}(\lambda)\right>$ of the edges $\left<e_1,\ldots,e_h\right>$ at which it touches the arcs of $\Lambda_1$.\footnote{In the sequel, we treat $\sigma(\lambda)$ as a circular sequence even if it has endpoints.} See Figure \ref{Fig:Sequence}.

\begin{figure}[htbp]
\begin{center}
\includegraphics[scale=0.8]{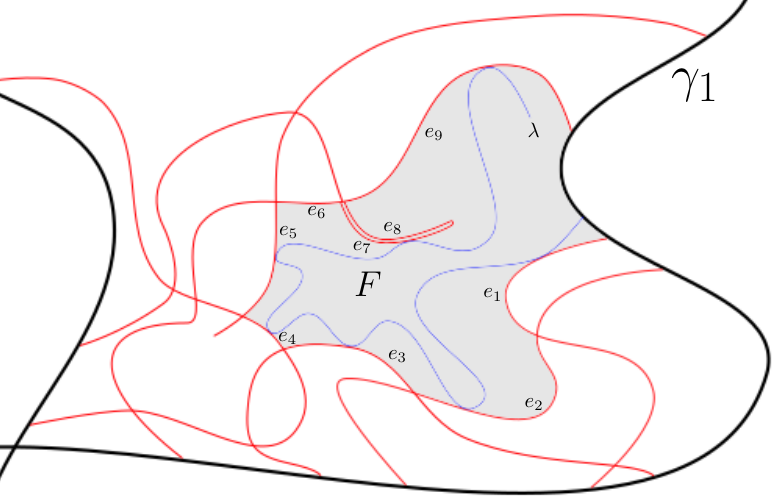}
\caption{\small Proof of Lemma \ref{Lemma:NoBiClique}. A cell $F\subset \Delta$ in the arrangement of $\Lambda_1\cup\{\gamma_1,\gamma_2\}$, for $|\Lambda_1|=m+5=7$, and an arc $\lambda\in \Lambda_F$. Notice that the arc $\lambda\in \Lambda_F$ determines the circular sub-sequence $\sigma(\lambda)=\left<e_1,e_2,e_3,e_4,e_5,e_7,e_9, e_1,\ldots\right>$.}
\label{Fig:Sequence}
\end{center}
\end{figure}

To complete the proof of Lemma \ref{Lemma:NoBiClique} we need the following property; Its somewhat weaker analogue was implicitly established by Salazar  for families of {\it closed} Jordan curves.

\begin{proposition}\label{Prop:CircularUnique}
Every arc $\lambda\in \Lambda_F$ has a unique circular sequence $\sigma(\lambda)$.
\end{proposition}

Though the proof of Proposition \ref{Prop:CircularUnique} overly follows the argument of Salazar \cite{Sa99}, the analysis must be adapted to the more general case of Jordan arcs.
We postpone the proof of Proposition \ref{Prop:CircularUnique} in the end of this section. 
\footnote{It is easy to check that the hypothesis that $|\Lambda_1|$ is $m+5$ (or larger) is essential for the correctness of Proposition 9.}

\medskip
By Proposition \ref{Prop:CircularUnique}, we have $z=|\Lambda_F|\leq m^{O(m)}$. Since the number of faces $F$ in $\A_\Delta$ is $O(m^3)$, we also have $l=|\Lambda_2|\leq m^{O(m)}$. This completes the proof of Lemma \ref{Lemma:NoBiClique}.
\end{proof}
	\bigskip

		\paragraph{Proof of Theorem \ref{Thm:Main1} -- wrap-up.}

By Lemma \ref{Lemma:GroundCurves}, the contact graph $\C(\Lambda)$ has $\Omega\left(\frac{T}{n^{\frac{2}{3m+15}}}\right)$ edges. On the other hand, combining Lemma \ref{Lemma:NoBiClique} with Theorem \ref{KST} yields an upper bound of $O\left(n^{2-\frac{1}{m+5}}\right)$ on the number of these edges. We thus obtain

$$
\frac{T}{n^{\frac{2}{3m+15}}}=O\left(n^{2-\frac{1}{m+5}}\right)
$$

\noindent or

$$
T = O\left(n^{2-\frac{1}{m+5} + \frac{2}{3m+15}}\right) =  O\left(n^{2-\frac{1}{3m+15}}\right).
$$

This concludes the proof of Theorem \ref{Thm:Main1}. $\Box$

\begin{figure}[htbp]
\begin{center}
\includegraphics[scale=0.7]{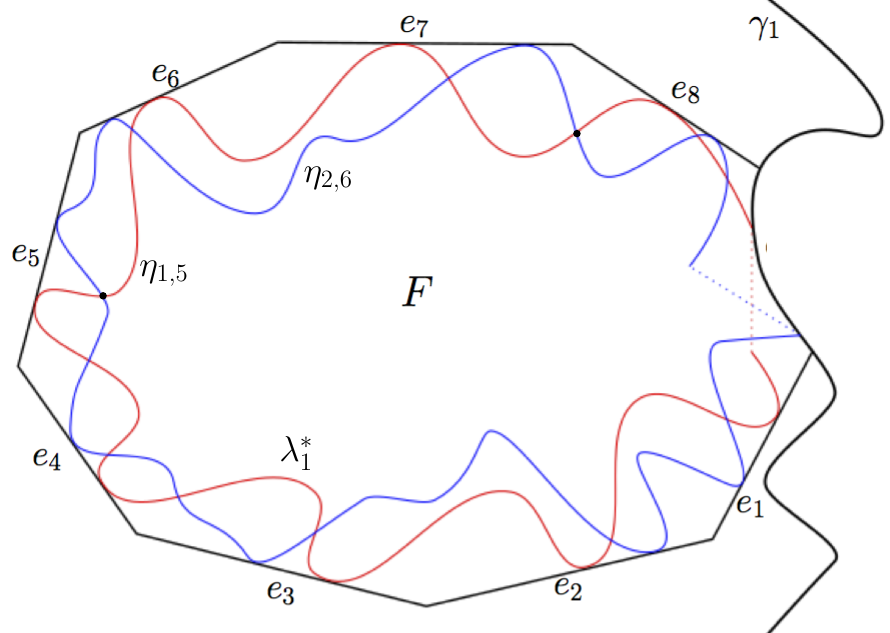}
\caption{\small Proof of Proposition \ref{Prop:CircularUnique}. The case $m=3$. A pair of arcs $\lambda_1$ (red) and $\lambda_2$ (blue) in $\Lambda_F$ that determine the same circular sequence $\sigma(\lambda_1)=\sigma(\lambda_2)=\left<e_1,e_2\ldots,e_8,e_1\ldots\right>$. We augment each arc to a closed Jordan curve, resp., $\lambda^*_1$ or $\lambda^*_2$, which lies within $F$. The edges $e_3,e_4,e_7,e_8,e_9$ are alt-edges, and the edges $e_1,e_2,e_5,e_6$ are hat-edges. The curves $\lambda^*_1$ and $\lambda^*_2$ intersect at least $8$ times, so the arcs $\lambda_1$ and $\lambda_2$ intersect at least $4$ times (contrary to the choice of $m$).}
\label{Fig:AltHat}
\end{center}
\end{figure}
				
\paragraph{Proof of Proposition \ref{Prop:CircularUnique}.}
	Assume by contradiction that there exists a pair of arcs $\lambda_1,\lambda_2 \in \Lambda_F$ that touch the edges of $\partial F$ in the same circular sequence 

$$
\sigma(\lambda_1)=\sigma(\lambda_2)=\left<e_{j_1}(\lambda_1),e_{j_2}(\lambda_1) ,\ldots,e_{j_{m+5}}(\lambda_1),e_{j_1}(\lambda_1) \right>. 
$$ 
For simplicity of notation, assume that $j_1 = 1, j_2 = 2, \ldots, j_{m+5} = m+5$, so $\sigma(\lambda)=\left<e_1,e_2,\ldots,e_{m+5}\right>$. We show that $\lambda_1$ and $\lambda_2$ must intersect at least $m+1$ times, which is contrary to our assumption that the original family $\Gamma$ (and, therefore, also $\Lambda_2$ and $\Lambda_F$) is $m$-intersecting.

To facilitate our analysis, we augment each arc $\lambda_i$, for $i\in \{1,2\}$, that is not already a closed Jordan curve, with an imaginary (open) arc which connects the endpoints of $\lambda_i$, lies entirely within $F$, and does not intersect the original arc $\lambda_i$. With some care, we can guarantee that the closed curves $\lambda_1^*$ and $\lambda_2^*$ do notiolate general position; that is, their imaginary portions intersect at finitely many points which do not coincide with the actual intersection points of the arcs of $\Lambda_F$.
Let $\lambda_i^*$ denote the resulting closed Jordan curve within $F$. Note that  each of the arcs $\lambda_i$, for $ i\in \{1,2\}$, meets every edge $e_k$ for $1\leq k \leq m+5$ at the unique point $p_{i,k}$.
Let $\eta_{i,k}$, for $i\in \{1,2\}$ and $1\leq k\leq m+5$, denote the subarc of $\lambda_i^*$ going from $p_{i,k}$ to $p_{i,(k+1)}$. 
Notice that, for each $i\in \{1,2\}$, the imaginary portion of $\lambda^*_i$  (if it exists) contains at most one of the above arcs $\eta_{i,k}$, for $1\leq k\leq m+5$.

We next assign every edge $e_k$, for $1\leq k\leq m+5$, to an intersection point of $\lambda^*_1$ and $\lambda^*_2$, so that no two of these edges are assigned to the same intersection. This will show that the curves $\lambda_1^*$ and $\lambda_2^*$ intersect at least $m+5$ times. 

Indeed, for each $1\leq i\leq m+5$ we find the points $p_{1,k}, p_{2,k}, p_{1,k+1},$ $p_{2,k+1}$  in one of the following orders along $\partial F$, up to switching the labels of $\lambda^*_{1}$ and $\lambda^*_2$: (a) $p_{1,k},p_{2,k}, p_{1,k+1}, p_{2,k+1}$ or (b) $p_{1,k},p_{2,k}, p_{2,k+1}, p_{1,k+1}$.\footnote{As before, the labeling of $e_1,\ldots,e_{m+5}$ is circular, that is, modulo $m+5$.} In the former case, we say that the edge $e_k$ is an {\it alt-edge,} and in the latter case $e_k$ is called a {\it hat-edge} ;  See Figure \ref{Fig:AltHat}.
If $e_k$ is an alt-edge, we assign it to an arbitrary point in the obviously non-empty intersection of $\eta_{1,k}$ and $\eta_{2,k}$. Otherwise, if $e_k$ is a hat-edge we assign $e_{k}$ to an arbitrary point
in the (again, non-empty) intersection $\eta_{1,k}\cap\eta_{2,k-1}$. It is immediate to check that any intersection point of $\lambda^*_1\cap \lambda^*_2$ is charged at most once. Namely, the points of $\eta_{1,k}\cap \eta_{2,k}$, for $1\leq k\leq m+5$, can be charged only by the alt-edge $e_k$, if it exists. In a similar manner, the points of $\eta_{1,k}\cap \eta_{2,k-1}$ and $\eta_{1,k-1}\cap \eta_{2,k}$ can be  
charged only by the hat-edge $e_k$ (again, if such an edge exists).

Notice that there can be at most $4$ indices $1\leq k\leq m+5$ for which the edge $e_k$ is assigned to an imaginary intersection point, outside the actual intersection $\lambda_1\cap \lambda_2$. This is because we can charge at most 2 intersections along each sub-arc $\eta_{i,k}$, for $i\in \{1,2\}$ and $1\leq k\leq m+5$, and each curve $\lambda^*_i$, for $i\in \{1,2\}$, can contain at most one sub-arc $\eta_{i,k}$ that does not fully lie within $\lambda_i$.

To conclude, we have shown that the Jordan curves $\lambda^*_1$ and $\lambda^*_2$ intersect at least $m+5$ times, and at least $m+1$ of these intersections belong to $\lambda_1\cap \lambda_2$. This contradiction to the $m$-intersecting property of $\Gamma$ proves Proposition \ref{Prop:CircularUnique}. $\Box$

						\section{Proof of Theorem \ref{Thm:Main2}}\label{Sec:Main2}
In this section we use Theorem \ref{Thm:Main1} to establish a more sensitive Crossing Lemma for Jordan arcs. To this end, we loosely follow the separator argument of Pach, Rubin and Tardos \cite[Section 3]{PRT18} with the main difference that we can now plug Theorem \ref{Thm:Main1} instead of a much weaker bound, which was used in \cite{PRT18} for families of Jordan curves with sufficiently dense intersection graphs. For the sake of completeness, we spell out all the technical details.

As before, let $\T$ denote the set of all the touching points that occur in a given $m$-intersecting family $\Gamma$ of $n$ Jordan arcs in general position and denote $T:=|\T|$. Let $\X$ denote the entire set of intersection points among the curves of $\Gamma$, and denote $X:=|\X|$. Recall that $\T\subseteq \X$, so we always have $T\leq X$. We also assume that $T\geq n$.

Let $d =\lfloor\frac{X}{n}\rfloor$. Consider the string graph $\I(\Gamma)$ determined by $\Gamma$, and note that the average degree of an arc in $\I(\Gamma)$ is proportional to $d$. 
We first reduce the degree of each curve in $\Gamma$ to at most $d$. To this end, we break each curve $\gamma \in \Gamma$ into sub-curves so that  all of them, have degree exactly $d$, with the possible exception of a single curve of degree at most $d$. We obtain a set of curves $\Gamma'$ with size $n' =\Theta(n)$, and notice that the sets $\T$ and $\X$ of, respectively, intersection points and touching points, remain unchanged.

 We repeatedly apply Theorem \ref{Thm:FoxPach} to the family $\Gamma'$ (or, more precisely, the intersection graph $\I\left(\Gamma')\right)$ until we obtain subsets of size less than the threshold $M:= C \frac{n^2d^3}{T^2}$ for a certain constant $C>1$. (Note that $M\gg d$ by Theorem \ref{Thm:Main1}, and it approaches $d$ if the ratio $X/T$ is small.) Consider all the different subsets that arise in the intermediate steps of the partition process. For any such subset of cardinality $k$, Theorem \ref{Thm:FoxPach} yields a separator of size $O\left(\sqrt{kd}\right)
$ which we add to the final separator set $S\subset \Gamma'$. Clearly, the subsets whose respective cardinalities $k$ belong the interval $ \left(\frac{3}{2}\right)^{i}M \leq k < \left(\frac{3}{2}\right)^{i+1}M$, for a fixed integer $0\leq i\leq \log_{\frac{3}{2}}n$, must be disjoint, so their number is at most $ \left(\frac{2}{3} \right)^i\frac{n}{M}$, and each of them contributes $O\left(\sqrt{\left(\frac{3}{2}\right)^i Md}\right)$ arcs to $S$. Thus, the overall contribution of such subsets to the size of $S$ is $O\left( \frac{n\sqrt{d}}{\sqrt{M}} \sqrt{\left(\frac{2}{3}\right)^i}\right)$. As we sum for all $1\leq i \leq \log_{\frac{3}{2}}n$, we obtain that $|S| = O\left(n\sqrt{d/M}\right)=O\left(T/d\right)$.
Note that the separator arcs of $S$ contain at most $|S|d=O(T)$ points of $\X$, and this quantity can be reduced to less than $T/2$ by choosing a sufficiently large constant $C$ in the definition of $M$. Therefore, we can assume, from now on, that the remaining curves of $\Gamma'\setminus S$ determine at least $|\T|-T/2\geq T/2$ touching points of $\T$.

						
Our construction easily implies that the terminal subsets of $\Gamma'$, to which we no longer apply Theorem \ref{Thm:FoxPach}, are pairwise disjoint.
Let $\{\Gamma_i| i\in I\}$ be the resulting final partition of $\Gamma' \setminus S$, so that no curve in $\Gamma_i$ intersects a curve in $\Gamma_j, j\neq i$.

It is immediate to check that $|I|\leq 2n'/M=O(n/M)$.
Therefore, the pigeonhole principle implies that there must be a subset $\Gamma_i\subset \Gamma'$ that determines at least $$
\frac{T}{2|I|}=\Omega\left(\frac{M T}{n}\right)=\Omega\left(\frac{nd^3}{T}\right)
$$

\noindent touching points of $\T$. Applying Theorem \ref{Thm:Main1} to this sub-family $\Gamma_i$ readily yields

\begin{equation}\label{Eq:Part}
\frac{nd^3}{T}=O\left(M^{2-\frac{1}{3m+15}}\right) 
\end{equation}

Denote $f:=X/T$, so that $f\geq 1$. Recall that we are to show that $f=\Omega\left(\left(T/n\right)^{\frac{1}{9m+45}}\right)$.
Using the definitions of $M$ and $f$, we can rewrite (\ref{Eq:Part})
as 

$$
\frac{M d}{f}=O\left(M^{2-\frac{1}{3m+15}}\right) 
$$
  
\noindent or

$$
\frac{d}{f}=O\left((f^2d)^{1-\frac{1}{3m+15}}\right). 
$$

Rearranging the last bound, and using the inequality $T\leq X$, we obtain

$$
\left(\frac{T}{n}\right)^{\frac{1}{3m+15}}=O\left(\left(\frac{X}{n}\right)^{\frac{1}{3m+15}}\right)= O\left(d^{\frac{1}{3m+15}}\right)=O(f^{3-\frac{2}{3m+15}})=O(f^3), 
$$

\noindent which yields the asserted lower bound for $f$. 
This concludes the proof of Theorem \ref{Thm:Main2}. $\Box$

\section{Conclusion and Open Problems}
Specializing to the significant case of $m$-intersecting Jordan curves, our Theorems \ref{Thm:Main1} and \ref{Sec:Main2} yield a distinct improvement over the bound provided by Theorem \ref{Thm:PRT}. It is achieved by replacing the intricate geometric charging scheme of Pach, Rubin and Tardos \cite{PRT18} with K\H{o}v\'{a}ri-S\'{o}s-Tur\'{a}n Theorem \ref{KST} -- a fairly standard tool from extremal combinatorics. This replacement relies on the $m$-intersecting property, which is essential to our proof of Lemma \ref{Lemma:NoBiClique} in Section \ref{Sec:Main1}.
However, if some pairs of Jordan curves can intersect an unbounded number of times, some families of $n$ such Jordan curves can determine $O(n^2)$ touching points while intersecting at most $O(n^2\log n)$ times\cite{FFPP10}. Hence, the general bound in Theorem \ref{Thm:PRT} cannot be improved beyond $\Omega\left(T \log \left(T/n\right)\right)$. It remains a significant open problem to determine the true growth rate of the smallest number of intersections that can arise in an arbitrary family of $n$ Jordan curves that determine $T$ touching points.

Another important problem, due to J. Pach, calls for the best possible bound in the special case of $1$-intersecting families of Jordan arcs.

\begin{conjecture}[J. Pach, 2017]
Let $\Gamma$ be a $1$-intersecting family of $n$ Jordan arcs in general position. Then $\Gamma$ determines $O(n)$ touching points.
\end{conjecture}

		\section{Acknowledgments}
		This work is based on my master's thesis, under the supervision of Dr. Natan Rubin. It was supported by grant No. 1452/15 from the Israel Science Foundation.

\end{document}